\documentclass[a4paper]{article}

\usepackage[english]{babel}
\usepackage[utf8x]{inputenc}
\usepackage[T1]{fontenc}
\usepackage{physics}

\usepackage[a4paper,top=3cm,bottom=2cm,left=3cm,right=3cm,marginparwidth=1.75cm]{geometry}

\usepackage{amsmath}
\usepackage{amssymb}
\usepackage{amsthm}
\usepackage{graphicx}
\usepackage{subcaption}
\usepackage{comment}
\usepackage[numbers]{natbib}
\usepackage[colorinlistoftodos]{todonotes}
\usepackage[colorlinks=true, allcolors=blue]{hyperref}
\newtheorem{theorem}{Theorem}

\newtheorem{proposition}{Proposition}
\newtheorem{lemma}{Lemma}
\newtheorem*{theorem*}{Theorem}
\newtheorem{definition}{Definition}
\usepackage[affil-it,noblocks]{authblk}

\newcommand*\Span{\text{span}}

\title{No quantum Ramsey theorem for stabilizer codes}
\author[1]{Yannis Bousba}
\author[2]{Travis B. Russell}
\affil[1]{Les \'{e}coles de Saint-Cyr Co\"{e}tquidan, Guer, France}
\affil[2]{United States Military Academy, West Point, NY USA}
\date{April, 2020 \\ Revised August, 2020}

\begin{document}
\maketitle

\begin{abstract} Recently, Nik Weaver proved a quantum analogue of the Ramsey theorem. Weaver's theorem states that for every positive integer $k$, there exists a positive integer $n_k$ such for any quantum channel on the $n_k \times n_k$ matrices, the corresponding quantum graph possesses either a $k$-dimensional quantum clique or a $k$-dimensional quantum anti-clique. Quantum anti-cliques coincide with error-correcting codes, while quantum cliques satisfy a dual property. In this paper we study the quantum graphs of mixed-unitary channels generated by tensor products of Pauli operators, which we call Pauli channels. We show that most quantum graphs arising from Pauli channels have non-trivial quantum cliques or quantum anti-cliques which are stabilizer codes. However, a reformulation of the quantum Ramsey theorem in terms of stabilizer codes and Pauli channels fails. Specifically, for every positive integer $n$, there exists an $n$-qubit Pauli channel for which any non-trivial quantum clique or quantum anti-clique fails to be a stabilizer code. We also show that this example is essentially unique, and hence most $n$-qubit Pauli channels have non-trivial quantum cliques or quantum anti-cliques which are stabilizer codes. \end{abstract}

\section{Introduction} The classical Ramsey theorems are a famous class of results originating from \cite{Ramsey1930} which demonstrate the phenomenon of discovering unexpected order in large, potentially chaotic, sets. One example of a Ramsey theorem can be phrased as follows: for every positive integer $k$, there exists a positive integer $n_k$ such that if $n_k$ people are found in a room, there exists a subset of $k$ people in the room such that either every member of the subset is acquainted with every other member, or no member of the subset is acquainted with any other member of the subset. In the first scenario, the subset is referred to as a $k$-clique, whereas in the second scenario, the subset is referred to as a $k$-anti-clique. In practice, $n_k$ is often much larger than $k$, and finding optimal bounds on the size of $n_k$ remains an active area of research \cite{conlon_fox_sudakov_2015}.

The classical Ramsey theorems have important corollaries in information theory and cryptography. Suppose a classical (probabilistic) channel is employed to encode a $n_k$-letter alphabet. The Ramsey theorem implies that there exists a $k$-letter subset of the alphabet for which either any message can be perfectly decoded with no probability of error or there is no guarantee that the even a subset of the message can be accurately decoded. These results are based on the equivalence of information theoretic properties of the channel with the combinatorial properties of the channel's confusability graph.

In \cite{DuanSeveriniWinter_2013}, Duan-Severini-Winter showed that the correspondence between classical channels and confusability graphs can be adapted to the setting of quantum channels. To this end, they show that there exists a correspondence between quantum channels and matrix systems - unital Hermitian subspaces of matrix algebras. Thus, they define a quantum graph to be a matrix system. They go on to prove that many properties of the quantum channel can be detected by studying the corresponding quantum graph.

In the paper \cite{Weaver2017}, Nik Weaver defines the notions of quantum cliques and quantum anti-cliques for a given quantum graph. Building upon his previous research in the theory of operator systems, he was able to prove a quantum analogue of the classical Ramsey theorem: namely that for each positive integer $k$ there exists an integer $n_k$ such that every matrix subsystem of the $n_k \times n_k$ matrices contains either a quantum $k$-clique or a quantum $k$-anti-clique.

In this paper we explore the possibility of a quantum Ramsey theorem for stabilizer codes. Stabilizer codes are examples of quantum codes which arise naturally in quantum computation. They are subspaces of the $n$-qubit Hilbert space whose error-correcting properties have been studied extensively (\cite{Shor95}, \cite{Gottesman96}, \cite{NielsenChuangTextbook}) as they are candidates for error-correcting codes that will be necessary to account for the noise which will arise in any physical implementation of the circuit model of quantum computation. 

We will show that, even under generous circumstances, a Ramsey theorem for stabilizer codes fails dramatically. Specifically, we study the quantum graphs corresponding to mixed unitary channels generated by unitaries taken from the Pauli group, which we call Pauli channels. We show that most quantum graphs of this form have non-trivial quantum cliques or quantum anti-cliques which are stabilizer codes. However, we also show that one can find, for any integer $n$, a quantum graph for a Pauli channel on the $2^n \times 2^n$ matrices with the property that any non-trivial quantum clique or quantum anti-clique fails to be a stabilizer code. We show, however, that this example is essentially unique.

Finally, we should emphasize that the notion of error-correction discussed throughout this paper is exact error-correction - i.e. correction of all errors. Thus, our negative results concerning the perfect error-correction properties of stabilizer codes and Pauli channels do not contradict the well-established asymptotic error-correction properties of stabilizer codes. Furthermore, our results say nothing about non-stabilizer codes in relation to Pauli channels, and hence do not contradict the main result of Weaver \cite{Weaver2017}.

Our paper is organized as follows. In section 2 we provide the basic definitions and notations used throughout the paper and recall any important results we will need from information theory, quantum theory and the literature on stabilizer codes. In section 3 we provide all results and their proofs.

\section{Preliminaries} 

Throughout this paper we will employ the following notation. We will use $\mathbb{N}$ to denote the set of positive integers. We let $\mathbb{R}$ denote the field of real numbers, we let $\mathbb{C}$ denote the field of complex numbers, and we let $\mathbb{F}_2$ denote the binary field $\{0,1\}$. Given a field $\mathbb{F}$ and an integer $n \in \mathbb{N}$, we let $\mathbb{F}^n$ denote the vector space of $n$-tuples with entries in $\mathbb{F}$. We let $M_n$ denote the set of $n \times n$ matrices with entries in $\mathbb{C}$, $M_{n,m}$ denote the $n \times m$ matrices, and we let $Tr(\cdot)$ denote the canonical trace function on $M_n$ defined by $Tr(A) = \sum_{i=1}^n a_{i,i}$ where $(a_{i,j})$ is the $(i,j)$ entry of the matrix for $A$ with respect to the canonical basis of $M_n$. 

Throughout we will freely use basic results from linear algebra and quantum theory. We refer the reader to \cite{NielsenChuangTextbook} for a good introduction the quantum theory. We will employ standard notation from quantum theory, except that we sometimes write $(h,k)$ for the inner product of vectors $h,k \in H$ for a given Hilbert space $H$. We will consider only finite-dimensional Hilbert spaces, and write $B(H)$ for the set of linear operators on a Hilbert space. We freely identify $B(H)$ with $M_n$ where $n = \dim(H)$.

\subsection{Classical channels and Ramsey's Theorem}

We begin by recalling the classical Ramsey Theorem from graph theory and describing its connection to information theory. 

\begin{definition}
Let $G$ be a simple undirected graph with vertex set $V$ and edge set $E$. A subset $C \subset V$ is called a \textbf{clique} if the subgraph generated by $C$ is complete - i.e., every pair of distinct vertices is connected by an edge. A subset $C$ of $V$ is called an \textbf{anti-clique} if the subgraph generated by $C$ is disconnected - i.e. no vertices of $C$ are connected by an edge. When $|C|=k$, we may refer to a $C$ as a $k$-clique if it is a clique or a $k$-anti-clique if it is an anti-clique.
\end{definition}

We will only consider the following simplified form of the Ramsey Theorem.

\begin{theorem}\label{thm: classical Ramsey theorem}[Ramsey \cite{Ramsey1930}]
For every $k \in \mathbb{N}$ there exists $n_k \in \mathbb{N}$ such that every undirected graph with $n_k$ vertices contains a $k$-clique or a $k$-anti-clique.
\end{theorem}

Classical channels are stochastic functions which which map letters in one alphabet to another. Let $n,m \in \mathbb{N}$ and suppose that $X=\{x_1, x_2, \dots, x_n\}$ is the input alphabet and $Y=\{y_1, y_2, \dots, y_m\}$ is the output alphabet. Then a \textbf{channel} $N: X \rightarrow Y$ is defined to be a matrix $(N(y_i|x_j))_{i \leq m, j \leq n}$ of positive real numbers satisfying the property that $\sum_{i=1}^m N(y_i|x_j) = 1$ for each $j \leq n$. The quantity $N(y_i|x_j)$ represents the probability that the symbol $x_j$ is encoded by the channel as $y_i$.

Given a channel $N$, we define the \textbf{confusability graph} of the channel to be the graph $G_{N}$ with vertex set $X$ and edge set defined by the relation $x_i \sim x_j$ if and only if $N(y_k|x_i)N(y_k|x_j) > 0$ for some $k \leq m$. Thus two vertices are connected by an edge if and only if it is possible that the channel maps both symbols to the same letter and hence confuses the letters.

By a \textbf{code}, we mean a non-empty subset $C$ of the input alphabet $X$. A code $C$ is called an \textbf{error-correcting code} if it is possible to determine from the output of $N$ which letters from $C$ were transmitted, provided it is known that only letters from $C$ were used as input. It is evident that $C$ is an error-correcting code for $N$ if and only if it impossible for two distinct letters from the code to be mapped to the same output. It follows that $C$ is an error-correcting code if and only if the corresponding set of vertices $C$ in the vertex set of $G_{N}$ constitutes an anti-clique. Conversely, we call a code $C$ a \textbf{private code} if it is not possible to distinguish any pair of letters from $C$ after application of the code, even if it is known that only letters from $C$ were used as input. Hence $C$ is a private code if for every pair $x_i,x_j \in C$ there exists a symbol $y_k \in Y$ such that $N(y_k|x_i)N(y_k|x_j) > 0$. Equivalently, a code $C$ is a private code if and only if the corresponding set of vertices $C$ in the vertex set of $G_N$ constitutes a clique.

The above discussion, together with Theorem \ref{thm: classical Ramsey theorem}, imply the following: for each $k \in \mathbb{N}$, if $N:X \rightarrow Y$ is a classical channel such that $|X| \geq n_k$, then there exists a code $C \subset X$ such that $|C|=k$ and $C$ is either an error-correcting code or a private code.

\subsection{Quantum channels and Weaver's Theorem}

In the vector state picture of quantum mechanics, a physical system is modeled by a Hilbert space $H$. The state of a quantum system is given by a unit vector $\ket{\phi} \in H$. Two states are considered equivalent if they differ by a phase - i.e. $\ket{\phi}$ is equivalent to $\ket{\psi}$ if $\ket{\phi} = e^{i\theta} \ket{\psi}$. In a closed system, the state evolves over time via unitary evolution $\ket{\phi} \mapsto U \ket{\phi}$ where $U$ is a unitary which does not depend on the state $\ket{\phi}$. In an open system, the evolution becomes stochastic in nature. Evolution of an open quantum system is modeled by a family of operators $\{E_1, E_2, \dots, E_m \in L(H)\}$ satisfying the completeness relation \[ \sum E_i^\dagger E_i = I. \] Then the evolution is modeled by $\ket{\psi} \mapsto t_i E_i \ket{\psi}$ with probability $\bra{\psi} E_i^\dagger E_i \ket{\psi}$ (where $t_i$ is a normalization constant). We refer to this stochastic mapping as a \textbf{quantum channel} (in the vector state picture).

Because of the stochastic nature of state evolution in an open quantum system, it is helpful to adopt a different notion of quantum state. Since two states are equivalent up to phase, we could consider the rank one projection $\ket{\psi}\bra{\psi}$ instead of the vector $\ket{\psi}$, eliminating the need to worry about phase. After evolution under a quantum channel, the state of the system could be any of $\{ \ket{\psi_1}\bra{\psi_1}, \ket{\psi_2}\bra{\psi_2}, \dots, \ket{\psi_m}\bra{\psi_m}\}$ with probability $\{p_1, p_2, \dots, p_m\}$, respectively, where each $p_i$ is positive and $\sum p_i = 1$. These constraints uniquely define a density operator $\rho = \sum p_i \ket{\psi_i} \bra{\psi_i}$. Conversely, any density operator $\rho$ can be decomposed in this form (although not uniquely). Thus, we may regard the state of the quantum system modeled by $H$ to be a density operator in $B(H)$. This leads to the following redefinition of quantum channel for density operators.

\begin{definition}\label{defn: quantum channel}[Quantum Channel]
Let $n,k \in \mathbb{N}$. A linear map $\mathcal{E}: M_n \rightarrow M_k$ is called a \textbf{quantum channel} if there exists operators $E_1, E_2, \dots, E_m \in M_{k,n}$ such that $\sum E_j^\dagger E_j = I_n$ and \[ \mathcal{E}(x) = \sum_{j=1}^m E_j x E_j^\dagger. \] The operators $\{E_1, E_2, \dots, E_m\}$ are called \textbf{noise operators} for the channel $\mathcal{E}$.
\end{definition}

Quantum channels can be equivalently defined as completely positive trace-preserving linear maps. These conditions ensure that if the input of the quantum channel is a density operator, then its output is also a density operator, and this property is stable under the tensor product operation.

Suppose $\mathcal{E}: M_n \to M_k$ is a quantum channel with noise operators $\{E_1, E_2, \dots, E_m\}$. Then Duan-Severini-Winter \cite{DuanSeveriniWinter_2013} define the \textbf{quantum graph} of the quantum channel $\mathcal{E}$ to be the vector space $G_{\mathcal{E}} := \Span \{E_i^\dagger E_j \} \subseteq M_n$. The quantum graph contains the identity operator since $\sum_j E_j^\dagger E_j = I_n$. It is also closed under the adjoint operation, since $(E_i^\dagger E_j)^\dagger = E_j^\dagger E_i$. Hence it is an example of a \textbf{matrix system}\footnote{In particular it is an operator system - a unital $\dagger$-closed subspace of $B(H)$ for some Hilbert space $H$.} - i.e. a unital $\dagger$-closed linear subspace of $M_n$.

To understand how the non-commutative graph of a quantum channel relates to the confusability graph of a classical channel, we need to introduce the notion of a quantum code. A \textbf{quantum code} is a linear subspace $C$ of a Hilbert space $H$. Let $P(C)$ be the unique orthogonal projection whose range is $C$. Given a matrix system $G \subseteq M_n$, a $k$-dimensional subspace $C \subset \mathbb{C}^n$ is called a \textbf{quantum clique} if $\dim(P(C)GP(C))=k^2$. It is called a \textbf{quantum anti-clique} if $\dim(P(C)GP(C)) = 1$. These properties ensure that the dimension of the vector space $P(C) G P(C)$ is maximal for a clique and minimal for an anti-clique.

Given a quantum code $C$, we say that a linear operator $x$ is \textbf{supported on $C$} if $x = PxP$ for $P=P(C)$. A quantum code $C \subseteq H$ is called a \textbf{quantum error-correcting code} for a quantum channel $\mathcal{E}: M_n \rightarrow M_k$ if there exists a quantum channel $\mathcal{F}: M_k \rightarrow M_n$ such that $\mathcal{F}(\mathcal{E}(x)) = x$ for all $x$ supported on $C$. The following theorem characterizes quantum error-correcting codes as quantum anti-cliques.

\begin{theorem}\label{thm: Knill-Laflamme Theorem}[Knill-Laflamme, \cite{KnillLaFlamme2000}]
Let $H$ be an $n$-dimensional Hilbert space, and let $C \subseteq H$ be a quantum code with $P=P(C)$. Then $C$ is a quantum error-correcting code for a quantum channel $\mathcal{E}$ if and only if $C$ is a quantum anti-clique for $G_{\mathcal{E}}$.
\end{theorem}

One method of proving Theorem \ref{thm: Knill-Laflamme Theorem} is to show that if $\mathcal{E}(x) = \sum E_j x E_j^\dagger$ then $C$ is error-correcting if and only if for every pair $\ket{h},\ket{k} \in C$ of orthogonal vectors and for every $i,j$ we have $( E_i \ket{h}, E_j \ket{k} )=0$. When this occurs, the vector $\ket{h}$ is mapped by the channel to a vector in $\Span_i \{ E_i \ket{h} \}$, while the vector $\ket{k}$ is mapped by the channel to the orthogonal space $\Span_i \{E_i \ket{k}\}$. Since quantum operations can distinguish between orthogonal vectors, it is possible to recover the originally transmitted states $\ket{h}$ and $\ket{k}$ up to phase.

We conclude this section by considering a notion of quantum private codes. Suppose $C \subseteq H$ is a quantum code and that $\mathcal{E}$ is a quantum channel with noise operators $\{E_1, E_2, \dots E_m\}$. Then we call $C$ a \textbf{quantum private code}\footnote{Our notion of quantum private code is more general than the one considered in \cite{KribsPloscker14}} for $\mathcal{E}$ if for every pair $\ket{a}, \ket{b} \in C$ of orthogonal vectors there exist $i,k \leq m$ such that $(E_i \ket{a}, E_j \ket{b}) \neq 0$. Hence, it is not possible to distinguish $\ket{a}$ from $\ket{b}$ after application of the channel $\mathcal{E}$ with certainty. The connection between quantum private codes and quantum cliques is illustrated in the following theorem, which to our knowledge does not appear in the literature.

\begin{theorem}
Let $H$ be an $n$-dimensional Hilbert space, and let $C \subseteq H$ be a quantum code with $P = P(C)$. Then $C$ is a quantum private code for $\mathcal{E}$ whenever $C$ is a quantum clique for $G_{\mathcal{E}}$.
\end{theorem}

\begin{proof}
We will regard $\mathcal{E}$ is a quantum channel in the state picture. Assume $\mathcal{E}$ has noise operators $\{E_1, E_2, \dots, E_m\}$ and consider $G_{\mathcal{E}} = \Span \{E_i^\dagger E_j\} \subseteq M_n$ and suppose that $C$ is a quantum clique for $G_{\mathcal{E}}$. Then $\dim(PG_{\mathcal{E}}P) = k^2$. Since $\dim(P M_n P) = k^2$, we see that for any operator $T \in M_n$ we have $PTP \in P G_{\mathcal{E}} P$. Now assume $\ket{a}, \ket{b} \in C$ and $\braket{a}{b}=0$. Let $T = \ket{a}\bra{b}$. Then since $PTP \in P G_{\mathcal{E}} P$, there exists a scalar matrix $(c_{i,j})_{i,j \leq m}$ such that $PTP = P( \sum c_{i,j} E_i^\dagger E_j) P$, and this matrix is necessarily non-zero. Now notice that \begin{eqnarray} (E_i \ket{a}, E_j \ket{b}) & = & \bra{a} E_j^\dagger E_i \ket{b} \nonumber \\ & = & \bra{a} P E_j^\dagger E_i P \ket{b}. \nonumber \end{eqnarray} If $(E_i \ket{a}, E_j \ket{b}) = 0$ for all $i,j$ it would follow that \[ 0 = \sum c_{i,j} (E_i \ket{a}, E_j \ket{b}) = \bra{a} P (\sum c_{i,j} E_i^\dagger E_j) P \ket{b} = \bra{a} T \ket{b} = 1, \] a contradiction. So it must be that $(E_i \ket{a}, E_j \ket{b}) \neq 0$ for some $i,j$.
\end{proof}

\subsection{Stabilizer codes and error correction}

Stabilizer codes are an important family of error-correcting codes in quantum computing. They are useful for analyzing separable quantum channels on $n$-qubit systems which model noise in a quantum circuit. The earliest example of a stabilizer code was discovered by Shor \cite{Shor95} and the general theory was developed by Gottesman \cite{Gottesman96}. See chapter 10 of \cite{NielsenChuangTextbook} for an excellent survey of this topic, including the results described below.

We first recall the $n$-qubit Pauli group $P_n$. With respect to the canonical basis of $\mathbb{C}^2$, we define matrices \[ X = \begin{pmatrix} 0 & 1 \\ 1 & 0 \end{pmatrix}, \quad Y = \begin{pmatrix} 0 & -i \\ i & 0 \end{pmatrix}, \quad Z = \begin{pmatrix} 1 & 0 \\ 0 & -1 \end{pmatrix}. \] These are the well-known Pauli matrices. We define the $n$-qubit Pauli group to be the finite group of matrices \[ P_n := \{ i^k \sigma_1 \otimes \sigma_2 \otimes \dots \otimes \sigma_n : k \in \{0,1,2,3\}, \sigma_l \in \{I,X,Y,Z\} \text{ for each } l \leq n \}. \]

\begin{definition}[Stabilizer]
A \textbf{stabilizer group} is a commutative subgroup $S \subset P_n$ (for some $n \in \mathbb{N}$) such that $-I \notin S$. A quantum code $C \subset \mathbb{C}^{2^n}$ is a \textbf{stabilizer code} if there exists a stabilizer group $S \subseteq P_n$ such that \[ C = \{ \ket{\phi} \in \mathbb{C}^{2^n} : g \ket{\phi} = \ket{\phi} \text{ for all } g \in S \}. \] In this case we say that $S$ is a \textbf{stabilizer} for $C$.
\end{definition}

The following useful properties of stabilizer codes are well-known.

\begin{proposition}[Properties of stabilizer codes] \label{prop: Properties of stabilizer codes}
Let $C \subset \mathbb{C}^{2^n}$ be a stabilizer code with stabilizer $S \subset P_n$. Then there exists $k \leq n$ such that $S$ is generated by independent elements $g_1, g_2, \dots, g_{n-k}$. In this case, $|S| = 2^{n-k}$ and $\dim(C) = 2^k$. Moreover, the orthogonal projection $P$ onto $C$ can be expressed as \[ P = \frac{1}{2^{n-k}} \prod_{i=1}^{n-k}(I + g_i) = \frac{1}{2^{n-k}}\sum_{g \in S} g. \]
\end{proposition}

We remark that every element of a stabilizer group $S$ is necessarily Hermitian. Indeed, if $g \in P_n$ and $g$ is not Hermitian, then $g^2 = -I$. However $-I \notin S$. In the context of stabilizer codes, we will be especially interested in a related class of mixed-unitary channels which we call \textbf{Pauli channels}.

\begin{definition}[Pauli channels]
By a \textbf{Pauli channel}, we mean a quantum channel $\mathcal{E}: M_{2^n} \rightarrow M_{2^n}$ of the form $x \mapsto \sum_{i=1}^m \lambda_i E_i x E_i^\dagger$ for some $E_1, E_2, \dots, E_m \in P_n$ and $\lambda_1, \lambda_2, \dots, \lambda_m > 0$ satisfying $\sum_i \lambda_i = 1$. 
\end{definition}

For convenience, we will ignore the scalars $\lambda$ and refer to $\{E_1, E_2, \dots, E_m\}$ as the noise operators for the Pauli channel $\mathcal{E}$ for the remainder of this paper. The next theorem characterizes the error-correcting stabilizer codes for a given Pauli channel.

\begin{theorem}[Gottesman, \cite{Gottesman96}]
Let $n \in \mathbb{N}$ and suppose that $\mathcal{E}: M_{2^n} \to M_{2^n}$ is a Pauli channel with noise operators $E_1,E_2,\dots,E_m \in P_n$. Then $C \subset \mathbb{C}^{2^n}$ is an error-correcting stabilizer code for $\mathcal{E}$ with stabilizer $S$ if and only if $E_i^\dagger E_j \notin Z(S) \setminus S$ for all $i,j \leq m$, where $Z(S)$ is the center of $S$ in $P_n$.
\end{theorem}

We conclude this section by recalling some techniques from the theory of stabilizer codes that will be useful. Let $n \in \mathbb{N}$. Then for vectors $\vec{a}, \vec{b} \in \mathbb{F}_2^n$, we define $X_{\vec{a}}:= X^{a_1} \otimes X^{a_2} \otimes \dots \otimes X^{a_n}$ and $Z_{\vec{b}} = Z^{b_1} \otimes Z^{b_2} \otimes \dots \otimes Z^{b_n}$. Since $ZX = -iY$ and $XZ = iY$, every element of $P_n$ can be written uniquely as $i^k X_{\vec{a}} Z_{\vec{b}}$ for some $k \in \{0,1,2,3\}$, $\vec{a},\vec{b} \in \mathbb{F}_2^n$. Given $g \in P_n$ with $g = i^k X_{\vec{a}} Z_{\vec{b}}$, we define its \textbf{check vector} to be the vector $r(g) := \vec{a} \oplus \vec{b} \in \mathbb{F}_2^{2n}$. It is easy to check that if $g,h \in P_n$ then $gh = hg$ if and only if \[ r(g)^T \begin{pmatrix} 0_n & I_n \\ I_n & 0_n \end{pmatrix} r(h) = \vec{0}. \] Equivalently, if $r(g) = \vec{a} \oplus \vec{b}$ and $r(h) = \vec{c} \oplus \vec{d}$, then $gh=hg$ if and only if $\langle \vec{a}, \vec{d} \rangle + \langle \vec{b}, \vec{c} \rangle = 0$, where the inner product is taken over the finite field $\mathbb{F}_2$. For convenience, we define the \textbf{twisted dot product} of two vectors $x= \vec{a} \oplus \vec{b}, y=\vec{c} \oplus \vec{d} \in \mathbb{F}_2^{2n}$ by \[ x * y := \langle \vec{a}, \vec{d} \rangle + \langle \vec{b}, \vec{c} \rangle. \] We summarize some properties of check vectors we will need in the following proposition.

\begin{proposition} \label{prop: properties of check vectors}
Let $n \in \mathbb{N}$. Then the following statements are true.
\begin{enumerate}
    \item For every $g,h \in P_n$, $g$ is a scalar multiple of $h$ if and only if $r(g) = r(h)$.
    \item For every $g,h \in P_n$, $gh=hg$ if and only if $r(g)*r(h)=0$. Otherwise $gh=-hg$.
    \item For every $g,h \in P_n$, $r(gh) = r(g)+r(h)$ and $r(g^\dagger) = r(g)$.
    \item A set $\{g_1, g_2, \dots, g_k\} \subseteq P_n$ of operators are independent as group elements of $P_n$ if and only if the set $\{r(g_1), r(g_2), \dots, r(g_k)\}$ is linearly independent in $\mathbb{F}_2^{2n}$.
\end{enumerate}
\end{proposition}

\section{Results} 

To arrive at our main result, we will need to study the dimension of $P G_{\mathcal{E}} P$ where $P$ is the orthogonal projection onto some stabilizer code and $\mathcal{E}$ is a Pauli channel. We begin by characterizing the possible values of $PgP$ whenever $g \in P_n$. Here and throughout this section, we write $g \sim h$ whenever $g,h \in P_n$ and $g$ is a scalar multiple of $h$ (i.e. $g = i^k h$ for some integer $k$).

We begin with some simple observations.

\begin{lemma} \label{lemma: non-zero trace}
Let $g \in P_n$. Then $Tr(g) \neq 0$ if and only if $g = i^k I$ for some integer $k$.
\end{lemma}

\begin{proof}
This is clear from the definition of $P_n$, since $Tr(X)=Tr(Y)=Tr(Z) = 0$ and $Tr(a \otimes b) = Tr(a)Tr(b)$.
\end{proof}

\begin{lemma} \label{lemma: Trace similar}
Let $g,h \in P_n$. Then $g \sim h$ if and only if $Tr(gh) \neq 0$.
\end{lemma}

\begin{proof}
By Lemma \ref{lemma: non-zero trace}, $Tr(gh) \neq 0$ if and only if $gh = i^k I$ for some integer $k$. Since $g^2 = \pm I$ for every $g \in P_n$ we see that $\pm h = i^k g$ if and only if $gh = i^k I$. The result follows.
\end{proof}

\begin{lemma} \label{lem: zero compression}
Let $S \subseteq P_n$ be a stabilizer group and let $P$ be the orthogonal projection onto the stabilizer code $C(S)$. Then for each $g \in P_n$, $PgP = 0$ if and only if $g \notin Z(S)$.
\end{lemma}

\begin{proof}
Suppose that $g \notin Z(S)$. Then $g$ anti-commutes with some non-trivial element $h$ of $S$. When this occurs we have $(I+h)g(I+h) = g + hg + gh + hgh = 0$ since $h$ is necessarily Hermitian and $h^2 = I$. Since we may assume any non-trivial element of $S$ is a generator of $S$, we see that $PgP = 0$ by the product form of $P$ in Proposition \ref{prop: Properties of stabilizer codes}. Now suppose that $g \in Z(S)$. Without loss of generality we may assume $g$ is Hermitian - otherwise consider $ig$. Then $Tr((PgP)^2) = Tr(PgP^2gP) = Tr(g^2P) = Tr(P) > 0$. It follows that $PgP \neq 0$.
\end{proof}

\begin{lemma} \label{lem: trace-orthogonal projections}
Let $S \subseteq P_n$ be a stabilizer group and let $P$ be the orthogonal projection onto the stabilizer code $C(S)$. Then for each $g, h \in Z(S)$, $PgP$ and $PhP$ are trace-orthogonal if and only if $gh \sim s$ for some $s \in S$.
\end{lemma}

\begin{proof}
By the summation form of $P$ in Proposition \ref{prop: Properties of stabilizer codes} we have \begin{eqnarray} Tr(PgPPhP) & = & Tr(ghP) \nonumber \\ & = & \sum_{s \in S} \frac{1}{2^{n-k}} Tr(ghs). \nonumber \end{eqnarray} Now for each $s \in S$, $Tr(ghs) \neq 0$ if and only if $gh \sim s$ by Lemma \ref{lemma: Trace similar}. Therefore if $gh \sim s$ is false for all $s \in S$ then $PgP$ and $PhP$ are trace orthogonal since $Tr(ghs)=0$ for all $s \in S$ in that case. On the other hand, suppose that $gh \sim s$ for some $s \in S$. Say $gh = i^m s$. Then \begin{eqnarray} \sum_{r \in S} \frac{1}{2^{n-k}} Tr(ghr) & = & \frac{i^m}{2^{n-k}} \sum_{r \in S} Tr(sr) \nonumber \\ & = & \frac{i^m}{2^{n-k}} \sum_{r \in S} Tr(s(sr)) \nonumber \\ & = & i^m Tr(P) \neq 0. \nonumber \end{eqnarray} The statement follows.
\end{proof}

\begin{definition} \label{defn: check vector image of set}
Let $W$ be any subset of $P_n$. We define \[ L(W) := \{ r(g) : g \in W \} \subseteq \mathbb{F}_2^{2n}. \]
\end{definition}

Using Definition \ref{defn: check vector image of set} and the lemmas above, we can prove the following characterization of the dimension of $PG_{\mathcal{E}}P$. Observe that when $S \subset P_n$ is a subgroup then $L(S)$ is a subspace of $\mathbb{F}_2^{2n}$ by part 3 of Proposition \ref{prop: properties of check vectors}.

\begin{theorem} \label{thm: dimension of PSP}
Let $\mathcal{E}: M_{2^n} \rightarrow M_{2^n}$ be a Pauli channel with noise operators $\{E_1, E_2, \dots, E_m \} \subseteq P_n$, and let $W_{\mathcal{E}} = \{ E_i^\dagger E_j\}$. Then for each stabilizer group $S \subset P_n$ we have \[ \dim(P G_{\mathcal{E}} P) = | \pi(L(W_{\mathcal{E}})) \cap L(Z(S))/L(S) | \] where $P$ is the projection onto the stabilizer code $C(S)$ and $\pi: \mathbb{F}_2^{2n} \rightarrow \mathbb{F}_2^{2n} / L(S)$ is the quotient map $\vec{a} \mapsto \vec{a} + L(S)$.
\end{theorem}

\begin{proof}
From part 1 of Proposition \ref{prop: properties of check vectors} we see that for each $g, h \in P_n$, $g \sim h$ if and only if $r(g) = r(h)$. Now suppose that $s=g_i g_j$ for some $s \in S$. Then by part 3 of Proposition \ref{prop: properties of check vectors}, $r(s) = r(g_i) + r(g_j)$ and hence $r(g_i) + L(S) = r(g_j) + L(S)$ in $\mathbb{F}_2^{2n} / L(S)$. Likewise, if $r(g_i) + S = r(g_j) + S$, then $r(g_i) + r(g_j) \in L(S)$ and hence $g_i g_j \sim s$ for some $s \in S$. It follows from Lemma \ref{lem: zero compression} and Lemma \ref{lem: trace-orthogonal projections} that the dimension of $PG_{\mathcal{E}}P$ is precisely the number of cosets of $L(Z(S))/L(S)$ present in the set $\pi(L(W_{\mathcal{E}}))$.
\end{proof}

We remark that $|L(Z(S))/L(S)| = 2^{\dim(L(Z(S))/L(S))} = 2^{\dim(L(Z(S))) - \dim(L(S))}$ for any stabilizer group $S$. Suppose that $S$ is a stabilizer group with independent generators $g_1, g_2, \dots, g_{n-k} \in P_n$. By part 4 of Proposition \ref{prop: properties of check vectors}, $\dim(L(S)) = n-k$. Hence to calculate $|L(Z(S))/L(S)|$ it remains to determine $\dim(L(Z(S)))$.

\begin{lemma} \label{lem: dimension of L(Z(S))}
Let $S \subset P_n$ be a stabilizer code with $n-k$ independent generators. Then \[ \dim(L(Z(S))) = n+k. \]
\end{lemma}

\begin{proof}
By part 2 of Proposition \ref{prop: properties of check vectors}, $g \in Z(S)$ if and only if $r(g) * r(s) = 0$ for all $s \in S$. Suppose that $g_1, g_2, \dots, g_{n-k}$ are generators for $S$. Then $r(g) * r(s) = 0$ for all $s \in S$ if and only if $r(g) * r(g_i) = 0$ for all $i \leq n-k$.

Define a linear operator $T: \mathbb{F}_2^{2n} \rightarrow \mathbb{F}_2^{n-k}$ via $T(\vec{a})_i = \vec{a} * r(g_i)$. Then $\ker(T) = L(Z(S))$. The operator $T$ can be represented by the matrix whose $i$-th row is given by $(\vec{b}_i \oplus \vec{a}_i)^T$ where $r(g_i) = \vec{a}_i \oplus \vec{b}_i$. Since the set $\{r(g_1), r(g_2), \dots, r(g_{n-k})\}$ is linearly independent, the rank of $T$ is $n-k$. By the rank-nullity Theorem, $\dim(L(Z(S))) = \dim(\ker(T)) = 2n - (n-k) = n+k$.
\end{proof}

We can now characterize the stabilizer codes which are quantum cliques for $G_{\mathcal{E}}$ for a given Pauli channel $\mathcal{E}$.

\begin{theorem} \label{thm: clique characterization}
Let $\mathcal{E}: M_{2^n} \rightarrow M_{2^n}$ be a Pauli channel with noise operators $\{E_1, E_2, \dots, E_m\} \subseteq P_m$, and let $W_{\mathcal{E}} = \{ E_i^\dagger E_j\}$. Then for each stabilizer group $S \subset P_n$ we have that $C(S)$ is a quantum clique if and only if \[  L(Z(S))/L(S) \subseteq \pi(L(W_{\mathcal{E}})) \] where $\pi: \mathbb{F}_2^{2n} \rightarrow \mathbb{F}_2^{2n} / L(S)$ is the quotient map $\vec{a} \mapsto \vec{a} + L(S)$.
\end{theorem}

\begin{proof}
By Proposition \ref{prop: Properties of stabilizer codes}, $\dim(C(S)) = 2^k$. Let $P$ be the projection onto $C(S)$. Then $C(S)$ is a quantum clique if and only $\dim(PG_{\mathcal{E}}P) = (2^k)^2 = 2^{2k}$. From Theorem \ref{thm: dimension of PSP}, we see that \[ \dim(P G_{\mathcal{E}} P) = | \pi(L(W_{\mathcal{E}})) \cap L(Z(S))/L(S) |. \] However $\dim(L(Z(S)) / L(S)) = (n+k) - (n-k) = 2k$ by Lemma \ref{lem: dimension of L(Z(S))}. Hence $| L(Z(S))/L(S)| = 2^{2k}$. We conclude that $\dim(P G_{\mathcal{E}} P) = 2^{2k}$ if and only if $L(Z(S))/L(S) \subseteq \pi(W_{\mathcal{E}})$.
\end{proof}

Having characterized the quantum cliques of a Pauli channel $\mathcal{E}$ which are stabilizer codes in terms of the set $W_{\mathcal{E}}$, we are almost ready to prove the main theorem. We will achieve this by demonstrating that for every $n$, there exists a Pauli channel $\mathcal{E}$ with no non-trivial quantum anti-cliques or quantum cliques which are stabilizer codes. In fact, we can construct an entire family of examples. To do this we need two more lemmas.

\begin{lemma} \label{lemma: basis for f_2^2n}
Assume that $\{h_1, h_2, \dots, h_n\} \subseteq P_n$ are commuting independent Hermitian operators. Then there exist commuting independent Hermitian operators $\{g_1, g_2, \dots, g_n\} \subseteq P_n$ such that for all $i \neq j$ we have $g_i h_i = -h_i g_i$ and $g_i h_j = h_j g_i$. Furthermore, $\{h_1, h_2, \dots, h_n, g_1, g_2, \dots, g_n\}$ is independent in $P_n$.
\end{lemma}

\begin{proof}
Let $\{h_1, h_2, \dots, h_n\} \subseteq P_n$ be commuting independent Hermitian operators. Then by Proposition \ref{prop: properties of check vectors}, the set $\{r(h_1), r(h_2), \dots, r(h_n)\}$ is linearly independent in $\mathbb{F}_2^{2n}$ and satisfies $r(h_i) * r(h_j) = 0$ for all $i,j \leq n$. 

Let $S$ be the stabilizer group generated by $\{h_1, h_2 \dots, h_n\}$. For each $l \in \{1,2, \dots, n \}$, let $S_l$ be the stabilizer group generated by $\{ h_1, h_2, \dots, h_n \} \setminus \{h_l\}$. By Lemma \ref{lem: dimension of L(Z(S))}, $\dim(L(Z(S_l))) = n+1$ for each $l \leq n$. Since $\{h_1, h_2, \dots, h_n\} \subseteq Z(S_l)$, and since $L(Z(S_l))$ is a subspace of $\mathbb{F}_2^{2n}$, there exists a basis of the form $\{r(h_1), r(h_2), \dots, r(h_n), r(g_l)\}$ for $L(Z(S_l))$, where $g_l$ is some Hermitian element of $P_n$. If $g_l h_l = h_l g_l$, then $r(g_l) \in L(Z(S))$ and hence $\{h_1, h_2, \dots, h_n, g_l\}$ is a linearly independent subset of $L(Z(S))$. However this is impossible since $\dim(L(Z(S)))=n$ by Lemma \ref{lem: dimension of L(Z(S))}. Thus $h_l g_l = -g_l h_l$. In this manner we obtain operators $g_1, g_2, \dots, g_n \in P_n$.

It may not be the case that the operators $\{g_1, g_2, \dots, g_n\}$ commute. If they do not commute, we will modify them so that they do as follows. Suppose that $g_1$ does not commute with all of $\{g_2, g_3, \dots, g_n\}$. Then whenever $g_k$ fails to commute with $g_1$, replace $g_k$ with $\hat{g}_k = h_1 g_k$. Then $g_1 \hat{g}_k = g_1 h_1 g_k = -h_1 g_1 g_k = h_1 g_k g_1 = \hat{g}_k g_1$. Furthermore, for each $l \neq k$ we have $\hat{g}_k h_l = h_l \hat{g}_k$ and $\hat{g}_k h_k = -h_k \hat{g}_k$. Letting $\hat{g}_i = g_i$ whenever $g_1$ commutes with $g_i$, we obtain the set $\{g_1, \hat{g}_2, \dots, \hat{g}_n\}$, which remains an independent set. Therefore, we may assume without loss of generality that $g_2, g_3, \dots, g_n$ all commute with $g_1$. Likewise, we may assume without loss of generality that $\{g_3, g_4, \dots, g_n\}$ all commute with $g_2$, $\{g_4, g_5, \dots, g_n\}$ all commute with $g_3$, and so on. Thus we obtain an independent commuting set of operators $\{g_1, g_2, \dots, g_n\}$.

Finally we must show that $\{h_1, h_2, \dots, h_n, g_1, g_2, \dots, g_n\}$ is independent in $P_n$. By part 4 of Proposition \ref{prop: properties of check vectors} it suffices to show that $\{r(h_1), r(h_2), \dots r(h_n), r(g_1), r(g_2), \dots, r(g_n)\}$ is a basis for $\mathbb{F}_2^{2n}$. For this it suffices to show that $\{r(h_1), r(h_2), \dots r(h_n), r(g_1), r(g_2), \dots, r(g_n)\}$ spans $\mathbb{F}_2^{2n}$. To this end, let $\vec{a} \in \mathbb{F}_2^{2n}$. Then $\vec{a} = r(g)$ for some Hermitian $g \in P_n$. If $g$ commutes with all of $\{h_1, h_2, \dots, h_n\}$ then $g \in Z(S)$. But then $r(g) \in L(Z(S))$. Since $\dim(L(Z(S)))=n$ by Lemma \ref{lem: dimension of L(Z(S))} and $\{r(h_1), r(h_2), \dots, r(h_n)\}$ is a basis for $L(Z(S))$, we have $r(g) \in \Span \{r(h_1), r(h_2), \dots, r(h_n)\}$ in this case. Now suppose that $g$ anti-commutes with $\{h_{k_1}, h_{k_2}, \dots, h_{k_l}\}$ for $k_1 < k_2 < \dots < k_l \leq n$ and that $g$ commutes with all other elements of $\{h_1, h_2, \dots, h_n\}$. Let $\vec{b} = \sum_{j=1}^l r(g_{k_j})$. Then $(r(g) + \vec{b})*r(h_i) = 0$ for all $i \leq n$. It follows from Proposition \ref{prop: properties of check vectors} that $r(g) + \vec{b} \in L(Z(S))$. Since $\vec{b} \in \Span \{r(g_1), r(g_2), \dots, r(g_n) \}$ and $L(Z(S))$ is spanned by $\{r(h_1), r(h_2), \dots, r(h_n)\}$, we must conclude that $r(g) = \vec{b} + (r(g) + \vec{b}) \in \Span \{r(h_1), r(h_2), \dots, r(h_n), r(g_1), r(g_2), \dots, r(g_n)\}$. \end{proof}

\begin{lemma} \label{lem: expand basis for stabilizer}
Let $\{h_1, h_2, \dots, h_{n-k}\} \subseteq P_n$ be a set of commuting Hermitian operators independent in $P_n$. Then there exist Hermitian operators $\{h_{n-k+1}, h_{n-k+2}, \dots, h_n\} \subseteq P_n$ such that $\{h_1, h_2, \dots, h_n\}$ is a set of commuting independent operators in $P_n$.
\end{lemma}

\begin{proof}
Let $l \in \{1, 2, \dots, k\}$ and assume that $\{h_1, h_2, \dots, h_{n-l}\}$ is an independent set of commuting operators in $P_n$. Let $S_l$ be the stabilizer group generated by $\{h_1, h_2, \dots, h_{n-l}\}$. By Lemma \ref{lem: dimension of L(Z(S))}, $\dim(L(Z(S_l))) = n+l > n$ and hence there exists a non-trivial Hermitian operator $h_{n-l+1} \in Z(S)$ such that $\{r(h_1), r(h_2), \dots, r(h_{n-l}), r(h_{n-l+1})\}$ is linearly independent. Consequently $\{h_1, h_2, \dots, h_{n-l}, h_{n-l+1}\}$ is a set of commuting independent operators in $P_n$ by part 4 of Proposition \ref{prop: properties of check vectors}. It follows that if $\{h_1, h_2, \dots, h_{n-k}\}$ are independent commuting operators generating a stabilizer group $S_k$ then there exists a chain of stabilizer groups $S_k \subset S_{k-1} \subset \dots \subset S_{1} \subset S_0$ with each $S_l$ generated by commuting Hermitian operators $\{h_1, h_2, \dots, h_{n-l}\}$ independent in $P_n$. The claim follows.
\end{proof}

We are now prepared to prove the main theorem.

\begin{theorem} \label{thm: Main theorem}
Let $n \in \mathbb{N}$. Let $S$ be a stabilizer group with $n$ independent generators. Define $\mathcal{E}: M_{2^n} \to M_{2^n}$ via $\mathcal{E}(x) = \sum_{h \in S} \lambda_h hxh$ where $\lambda_h > 0$ for each $h \in S$ and $\sum_{h \in S} \lambda_h = 1$. Then $G_{\mathcal{E}}$ has no non-trivial quantum cliques or quantum anti-cliques which are stabilizer codes.
\end{theorem}

Before giving the proof we remark that it is easy to find examples of channels like the one described in the theorem. For instance, take $S$ to be the stabilizer group in $P_n$ with generators $\{X \otimes I \otimes \dots \otimes I, I \otimes X \otimes I \otimes \dots \otimes I, \dots, I \otimes I \otimes \dots \otimes X \}$.

\begin{proof}
Let $R$ be a stabilizer group with $n-k$ generators, where $k < n$. We will show that $C(R)$ is neither a clique nor an anti-clique for $G_{\mathcal{E}}$.

We first prove that $C(R)$ is not a quantum anti-clique for $G_{\mathcal{E}}$. To do this, we will show that \[ |\pi(L(S)) \cap L(Z(R)) / L(R)| > 1 \] where $\pi: \mathbb{F}_2^{2n} \to \mathbb{F}_2^{2n} / L(R)$ is the quotient map. The claim will follow by Theorem \ref{thm: dimension of PSP}. To do this, it suffices find a non-trivial $h \in S$ such that $r(h) \in L(Z(R)) \setminus L(R)$.

By Lemma \ref{lem: dimension of L(Z(S))} we see that $\dim(L(Z(R))) = n + k$. Since $\dim(L(S) \cap L(Z(R))) = \dim(L(S)) + \dim(L(Z(R))) - \dim(L(S) + L(Z(R))) \geq k$, we conclude that there exist $h_1, h_2, \dots, h_l \in S$ with $l \geq k$ such that $\{r(h_1), r(h_2), \dots, r(h_l)\} \subseteq L(S) \cap L(Z(R))$ is linearly independent. If $V := \Span \{r(h_1), r(h_2), \dots, r(h_l)\}$ is not a subspace of $L(R)$ then there exists a non-trivial $h \in S$ such that $r(h) \in V \setminus L(R)$ and hence $r(h) \in L(Z(R)) \cap L(S) \setminus L(R)$. Therefore we must consider the case $V \subseteq L(R)$ and hence $V = L(S) \cap L(R)$. If $V = L(R)$, then because $\dim(L(R)) < \dim(L(S))$ there exists a non-trivial $h \in S$ such that $r(h) \in L(S) \setminus L(R)$. But because $L(R) \subseteq L(S)$ and $r(g) * r(h) = 0$ for all $g \in S$, we must conclude that $r(g) * r(h) = 0$ for all $g \in R$ and hence $r(h) \in L(Z(R)) \cap L(S) \setminus L(R)$. If $V$ is a proper subspace of $L(R)$, then we may choose $\{w_{l+1}, w_{l+2}, \dots, w_{n-k}\} \subseteq R$ such that $\{r(h_1), r(h_2), \dots r(h_l), r(w_{l+1}), r(w_{l+2}), \dots r(w_{n-k})\}$ is a basis for $L(R)$. Let $R'$ be the stabilizer group generated by $\{w_{l+1}, w_{l+2}, \dots, w_{n-k}\}$. Then by Lemma \ref{lem: dimension of L(Z(S))} $\dim(L(Z(R'))) = n + k + l$. Hence $\dim(L(S) \cap L(Z(R'))) \geq k+l$. It follows that there exists $h \in S$ such that $r(h) \notin V$ and $r(h) \in L(S) \cap L(Z(R'))$. Since $L(R) = \Span \{r(h_1), r(h_2), \dots r(h_l), r(w_{l+1}), r(w_{l+2}), \dots r(w_{n-k})\}$ and $V = \Span \{r(h_1), r(h_2), \dots r(h_l)\} \subseteq L(S)$ we see that $r(h) * r(g) =  0$ for all $g \in R$. We deduce that $r(h) \in L(Z(R)) \cap L(S) \setminus L(R)$. Therefore we conclude that $C(R)$ is not a quantum anti-clique for $G_{\mathcal{E}}$.

Finally we must show that $C(R)$ is not a clique. Let $\{w_1, w_2, \dots, w_{n-k}\}$ be an independent set of generators for $R$. Then there exist Hermitian operators $\{w_{n-k+1}, \dots, w_n\}$ such that $\{w_1, w_2, \dots, w_n\}$ is an independent set of commuting operators in $P_n$ by Lemma \ref{lem: expand basis for stabilizer}. By Lemma \ref{lemma: basis for f_2^2n}, we can find Hermitian operators $\{v_1, v_2, \dots, v_n\}$ such that $w_i v_i = -v_i w_i$ and $w_i v_j = v_j w_i$ for all $i \neq j$ and such that \[ \{r(w_1), r(w_2), \dots, r(w_n), r(v_1), r(v_2), \dots, r(v_n)\} \] is a basis for $\mathbb{F}_2^{2n}$. Since $L(R) = \Span \{r(w_1), r(w_2), \dots, r(w_{n-k})\}$, it is evident that \[ \{ r(w_{n-k+1}) + L(R), \dots, r(w_n)+L(R), r(v_{n-k+1}) + L(R), \dots, r(v_n) + L(R)\} \] is a basis for the quotient vector space $L(Z(R)) / L(R)$. By Theorem \ref{thm: clique characterization}, $C(R)$ is a clique for $G_{\mathcal{E}}$ if and only if $L(Z(R))/L(R) \subseteq L(S)$. However this is impossible. Indeed, suppose that $\vec{a} \in r(w_n) + L(R)$ and $\vec{b} \in r(v_n) + L(R)$. We may assume that $\vec{a} = r(w_n) + r_1$ and $\vec{b} = r(v_n) + r_2$ for some $r_1,r_2 \in L(R)$. Then \begin{eqnarray} \vec{a} * \vec{b} & = & (r(w_n) + r_1) * (r(v_n) + r_2) \nonumber \\ & = & r(w_n) * r(v_n) + r(w_n) * r_2 + r_1 * r(v_n) + r_1 * r_2 \nonumber \\ & = & 1. \nonumber \end{eqnarray} However if $\vec{a}, \vec{b} \in L(S)$ then $\vec{a} * \vec{b} = 0$ since the elements of $S$ commute, by part 2 of Proposition \ref{prop: properties of check vectors}. It follows that $L(Z(R))/L(R)$ is not a subset of $\pi(L(S))$ and hence $C(R)$ is not a quantum clique by Theorem \ref{thm: clique characterization}.
\end{proof}

We conclude by showing that the quantum graphs considered in Theorem \ref{thm: Main theorem} are the only quantum graphs for Pauli channels lacking non-trivial quantum cliques or anti-cliques from the set of stabilizer codes.

\begin{theorem} \label{thm: only one example}
Let $n \in \mathbb{N}$. Suppose that $\mathcal{F}:M_{2^n} \to M_{2^n}$ is a Pauli channel. Then one of the following hold. \begin{enumerate}
    \item There exists a stabilizer group $R$ such that $C(R)$ is a non-trivial quantum anti-clique for $G_{\mathcal{F}}$.
    \item There exists a stabilizer group $R$ such that $C(R)$ is a non-trivial quantum clique for $G_{\mathcal{F}}$.
    \item There exists a Pauli channel $\mathcal{E}$ of the form described in Theorem \ref{thm: Main theorem} such that $G_{\mathcal{F}} = G_{\mathcal{E}}$.
\end{enumerate}
\end{theorem}

\begin{proof}
Suppose that $G_{\mathcal{F}} \neq G_{\mathcal{E}}$ for any $\mathcal{E}$ of the form described in Theorem \ref{thm: Main theorem}. Since $\mathcal{F}$ is a Pauli channel, there exist $\{E_1, E_2, \dots, E_m\} \subseteq P_n$ such that $\mathcal{F}(x) = \sum_{i=1}^m \lambda_i E_i x E_i^\dagger$. Without loss of generality we may assume all of $\{E_1, E_2, \dots, E_m\}$ are all Hermitian. Indeed, if $E_k$ is not Hermitian we can replace it with the Hermitian operator $F_k := iE_k$ since $E_k x E_k^\dagger = F_k x F_k$.

Let us first assume that the operators $\{E_1, E_2, \dots, E_m\}$ commute. Without loss of generality, we may assume that $-I$ is not in the subgroup $S$ generated by $\{E_1, E_2, \dots, E_m\}$ in $P_n$. Indeed, let $\{h_1, h_2, \dots, h_l\} \subseteq \{E_1, E_2, \dots, E_m\}$ be independent operators which generate $S$. If $-I = h_1^{\alpha_1} h_2^{\alpha_2} \dots h_l^{\alpha_l}$ for some $\alpha_1, \alpha_2, \dots, \alpha_l \in \mathbb{F}_2$, then for any $\alpha_k \neq 0$ we have $-h_k = h_k h_1^{\alpha_1} \dots h_l^{\alpha_l}$. But we can replace $h_k$ with $F_k := -h_k$ without changing the map $\mathcal{F}$ since $h_k x h_k = F_k x F_k$. Thus we may assume that the group $S$ generated by $\{E_1, E_2, \dots, E_m\}$ is a stabilizer group. By Lemma \ref{lem: expand basis for stabilizer} we can find $\{h_{l+1}, \dots, h_n\}$ such that $\{h_1, h_2, \dots, h_n\}$ is an independent set of commuting Hermitian operators. Let $S'$ be the group generated by these $n$ generators. Since $G_{\mathcal{F}}$ cannot equal $G_{\mathcal{E}}$ where $\mathcal{E}(x) = \sum_{h \in S'} hxh$, we conclude that $G_{\mathcal{F}}$ is a proper subspace of $G_{\mathcal{E}}$. In particular, $L(W_{\mathcal{F}})$ is a proper subset of $L(S')$ where $W_{\mathcal{F}} = \{E_i^\dagger E_j\}$. Without loss of generality, we may assume that $L(W_{\mathcal{F}}) \subseteq L(S') \setminus \{h_n\}$. Let $g_1, g_2, \dots, g_n$ be Hermitian operators with the properties described in Lemma \ref{lemma: basis for f_2^2n}. Let $R$ be the stabilizer group generated by $\{g_1, g_2, \dots, g_{n-1}\}$. Then $R$ satisfies $L(Z(R)) \cap (L(S') \setminus \{h_n\}) = \{\vec{0}\}=\{r(I)\}$. Indeed, every element of $S'$ fails to commute with at least one of the operators $\{g_1, g_2, \dots, g_{n-1}\}$ except for $h_n$ and $I$. It follows that $L(Z(R)) \cap L(W_{\mathcal{F}}) = \{r(I)\}$. So $C(R)$ is a non-trivial quantum anti-clique for $G_{\mathcal{F}}$ by Theorem \ref{thm: dimension of PSP}.

Finally assume that $\{E_1, E_2, \dots, E_m\}$ do not all commute. Hence we may assume that $h=E_i$ and $g=E_j$ do not commute for some $i \neq j$. By Lemma \ref{lem: expand basis for stabilizer} there exist Hermitian $\{h_2, h_3, \dots, h_n\}$ such that $\{h, h_2, \dots, h_n\}$ is a set of independent commuting operators in $P_n$. Moreover, we may assume without loss of generality that $h_i g = g h_i$ for each $i > 1$. Indeed, if $h_i g = -g h_i$ for some $i$, we can replace $h_i$ with $\hat{h}_i := h_i h$ to get $\hat{h}_i g = g \hat{h}_i$ without affecting the independence of the set $\{h, h_2, \dots, h_n\}$. Let $R$ be the stabilizer group generated by $\{h_2, h_3, \dots, h_n\}$. Then $\{r(I), r(g), r(h), r(g) + r(h)\} \subseteq L(W_{\mathcal{F}}) \cap L(Z(R))$. Moreover, $r(I), r(h), r(g)$ and $r(h) + r(g)$ belong to different cosets of $L(Z(R))/L(R)$. Indeed, $r(h), r(g), r(g) + r(h) \notin L(R)$ whereas $r(I) \in L(R)$. Since $r(g) + r(h) \notin L(R)$, $r(g)$ and $r(h)$ belong to different cosets. Since $r(g) + (r(h) + r(g)) = r(h) \notin L(R)$, $r(g)$ and $r(h) + r(g)$ belong to different cosets, and similarly $r(h)$ and $r(h) + r(g)$ belong to different cosets. Therefore $|\pi(L(W_{\mathcal{F}})) \cap \pi(L(Z(R)))| \geq 4$. But $|\pi(L(Z(R)))| = 4$ since $|\pi(L(Z(R)))| = 2^{\dim(L(Z(R))-\dim(L(R))}$, $\dim(L(Z(R)))=n+1$ by Lemma \ref{lem: dimension of L(Z(S))} and $\dim(L(R))=n-1$. We conclude that $C(R)$ is a non-trivial quantum clique for $G_{\mathcal{F}}$ by Theorem \ref{thm: clique characterization}.
\end{proof}

\section*{Acknowledgments}

This paper builds upon the Master's thesis of the first author \cite{BousbaThesis} which was completed under the supervision of the second author. We thank Professor Guy Chasse of Les \'{e}coles de Saint-Cyr Co\"{e}tquidan and Professor Tina Hartley of the United States Military Academy for arranging for the first author's visit to the United States Military Academy where this research took place. We also thank the editor and referees for their careful reading and comments which improved the exposition of this work.

\bibliographystyle{ieeetr}
\bibliography{references}
\end{document}